\documentclass[runningheads]{lipics-v2019}

\usepackage{microtype}
\usepackage{amsmath}
\usepackage{graphicx}
\usepackage{xcolor}
\usepackage{soul}
\usepackage{algorithmic}
\usepackage{algorithm}
\usepackage{tabto}
\usepackage{caption}
\usepackage{subcaption}
\newcommand{\remove}[1]{}
\usepackage{graphicx}

\newcommand{\facilityset}{\mathcal{F}}	
	
\newcommand{\facset}{\mathcal{F}}				

\newcommand{\clientset}{\mathcal{C}}			
\newcommand{\cliset}{\mathcal{C}}				

\newcommand{\capacity}{\textit{u}}				

\newcommand{\crich}{\mathcal{C}_B}				
\newcommand{\cdense}{\mathcal{C}_B}				
\newcommand{\csparse}{\mathcal{C}_S}

\newcommand{\ballofj}[1]{\mathcal{B}_{#1}} 	

\newcommand{\C}[1]{\hat{C_{#1}}}					

\newcommand{\bundle}[1]{\mathcal{F}_{#1}}			
\newcommand{\neighbor}[1]{\mathcal{F}_{#1}}

\newcommand{\T}[1]{\xi({#1})}

\newcommand{\cen}[1]{B_{#1}}

\newcommand{\dist}[2]{c(#1,~#2)}

\newcommand{\facilitycost}{\textit{$ f_i $}}		

\newcommand{\bard}[1]{{\Delta_{#1}}}			

\newcommand{\floor}[1]{\lfloor{#1}\rfloor}

\newcommand{\sumlimits}[2]{\displaystyle\sum\limits_{#1}^{#2}}												

\newcommand{\sumap}[1]{\sumlimits{#1}{}} 

\newcommand{\load}[1]{d_{#1}}

\newcommand{\singleclient}{\textit{j}}
\newcommand{\singlefacility}{\textit{i}}

\newcommand{\zofi}[1]{w_#1}

\newcommand{\primezofi}[1]{w'_{#1}}

\newcommand{\demandofj}[1]{\textit{$ \Delta_{#1} $}}

\newcommand{\loadjinc}[2]{\phi(#1,~#2)}

\newcommand{\etal}{\textit{et al}.}

\newcommand{\dense}{j_b}

\newcommand{\floordjbyu}[1]{\lfloor{\Delta_{#1}/u}\rfloor}

\newcommand{\sigmaone}{p}

\newcommand{\Sr}[1]{S_{#1}}

\newcommand{\clirem}{\mathcal{C}_r}
\newcommand{\clipen}{\mathcal{C}_p}
\newcommand{\cliout}{\mathcal{C}_o}
\newcommand{\opt}[1]{LP_{opt}}
\newcommand{\lp}[1]{LP_{#1}}

\newcommand{\Bundler}{B_r}
\newcommand{\bundleone}{B_{r}^b}
\newcommand{\bundletwo}{B_{r}^s}

\newcommand{\MC}{bundle }
\newcommand{\MCs}{bundles }

\newcommand{\iofj}{i_{j'}}

\title{Constant Approximation for Capacitated Facility Location Problems with Penalties/Outliers}

\titlerunning{Capacitated Facility Location Problems with Penalties/Outliers}

\author{Rajni Dabas}{Department of Computer Science, University of Delhi, India}{rajni@cs.du.ac.in}{}{[Supported by a UGC-JRF]}
\author{Neelima Gupta}{Department of Computer Science, University of Delhi, India}{ngupta@cs.du.ac.in}{}{}

\authorrunning{Dabas and Gupta}

\Copyright{Neelima Gupta and Rajni Dabas}

\nolinenumbers
\begin{document}
\maketitle             

\begin{abstract}

In this paper, we present a framework to design approximation algorithms for capacitated facility location problems with penalties/outliers by rounding a  solution to the standard LP. Primal-dual technique, which has been particularly successful in dealing with outliers and penalties, has not been very successful in dealing with capacities. For example, despite unbounded integrality gap for facility location problem with outliers(FLPO), Charikar \etal~\cite{charikar2001algorithms} were able to get around it by guessing the maximum facility opening cost in the optimal and provide a primal-dual solution for the problem. On the other hand, no primal-dual solution has been able to break the integrality gap of the capacitated facility location problem(CFLP). (Standard)LP-Rounding techniques had also not been very successful in dealing with capacities until a recent work by Grover~\etal~\cite{GroverGKP18}. Their constant factor approximation violating the capacities by a small factor ($1 + \epsilon$) is promising while dealing with capacities. Though LP-rounding has not been very promising while dealing with penalties and outliers, we successfully apply it to deal with them along with capacities. Solutions obtained by using LP-rounding techniques are interesting as they are easy to integrate with other LP-based algorithms. 
 
We apply our framework to obtain first constant factor approximations for capacitated facility location problem with outlier (CFLPO) and capacitated $k$-facility location problem with penalty (C$k$FLPP) for hard uniform capacities. Our solutions incur slight violations in capacities, ($1 + \epsilon$) for the problems without cardinality($k$) constraint and ($2 + \epsilon$) for the problems with the cardinality constraint. For the outlier variant, we also incur a small loss ($1 + \epsilon$) in outliers. Due to the unbounded integrality gaps in the underlying problems, the violations are inevitable while rounding solutions to the standard LP. Thus we achieve the best possible by rounding the solution of natural LP for these problems. To the best of our knowledge, no results are known for CFLPO and C$k$FLPP. The only result known for CFLPP uses local search.

As a byproduct, we obtain first constant factor approximations for ($i$) Capacitated $k$-Median with penalties(C$k$MP) ($ii$) the uncapacitated variants of FLPO, $k$MP and $k$FLPP using LP-rounding.

\keywords{ Facility Location  \and Outliers \and Penalties \and Approximation.}
\end{abstract}

\section{Introduction}
The facility location problem is a fundamental and well studied problem in operational research and theoretical computer science~\cite{Shmoys,ChudakS03,Jain:2001,Byrka07,Chudak98,Li13}. In (uncapacitated) facility location problem (FLP), we are given a set $\cliset$ of $m$ clients and a set $\facilityset$ of $n$ facility locations. Setting up a facility at location $i$ incurs cost $f_i$(called the {\em facility opening cost} or simply the {\em facility cost} ) and servicing a client $j$ by a facility $i$ incurs cost $\dist{i}{j}$ (called the {\em service cost}). We assume that the costs are metric, i.e., they satisfy the triangle inequality. The objective is to select a set $\facilityset' \subseteq \facilityset$, so that the total cost of opening the facilities in $\facilityset'$ and the cost of servicing all the clients by opened facilities is minimized. If the maximum number of clients a facility $i$ can serve is bounded by $u_i$, the problem is called {\em capacitated facility location problem}. 
The above formulation of the problem is sensitive towards noisy clients in the sense that few distantly located clients, called  {\em outliers} can significantly increase the cost of the solution.
To address the concern, Charikar \etal~\cite{charikar2001algorithms} defined the problems of $k-$center, $k-$median and facility location with outliers. 
In {\em facility location problem with outliers} (FLPO), we are given a bound $t$ 
on the maximum number of clients that can be considered as outliers. The objective now is to identify the locations to install facilities and serve at least $m-t$ clients, so that the total cost for setting up the facilities and servicing the selected clients is minimized. In another closely related variant of the problem, called the {\em facility location problem with penalties} (FLPP), instead of a hard bound on the number of outliers, 
we are allowed to leave a client $j \in \cliset$ unserved by paying a penalty cost $p_j$. The objective, then, is to identify the locations to install facilities and select the clients to be served, so that the total cost for setting up the facilities, servicing the selected clients and paying the penalty cost of the unserved clients is minimized. 
 In this paper, we study capacitated facility location problems with penalties/outliers(CFLPP/CFLPO) when the capacities are uniform i.e. $u_i = u$ for all $i \in \facilityset$. We  present a framework to obtain constant factor approximation for these problems. 
We present first constant factor approximation
for CFLPO
violating both the capacities as well as the outliers by a small factor of $(1 + \epsilon)$ each. The result is obtained by rounding a solution to the standard LP; the violations are inevitable as both CFLP~\cite{Shmoys} as well as FLPO ~\cite{charikar2001algorithms} are known to have unbounded integrality gaps from the standard LPs. To the best of our knowledge, CFLPO 
has not been studied earlier. In particular, we present the following result:

\begin{theorem}
\label{thm-CFLPO}
There is a polynomial time algorithm that approximates 
uniform capacitated facility location problem with outliers within a constant factor $(O(1/\epsilon^2))$ violating the capacities by a factor of at most $(1 + \epsilon)$ leaving at most $(1 + \epsilon) t$ outliers, for $\epsilon>0$.
\end{theorem}

We next present our result for uniform {\em capacitated facility location problem with penalties} (CFLPP). 
Though there has been significant amount of work on uncapacitated FLPP~\cite{charikar2001algorithms,jain2003greedy,xu2005lp,xu2009improved,LiDXX15,CFLPP_Byrka} with the current best being $1.8526$ due to Xu and Xu~\cite{xu2009improved}, the only result known for the capacitated case is due to Gupta \etal~\cite{GuptaG14} using local search. In this paper, we present first constant factor approximation for the problem by rounding the solution to standard LP with slight violation in capacities, which is inevitable due to integrality gap of CFLP~\cite{Shmoys} using standard LP. Though our result is weaker than the result
of Gupta \etal~
it is interesting as it is obtained by rounding solution to standard LP. LP-rounding based solutions are interesting as they are easier to integrate with other LP-based solutions. One direct application is presented in this paper by extending the result to a more generalized problem with an additional cardinality constraint. The result is stated in Theorem \ref{thm-ckflpp1}. It is not clear how the result of Gupta \etal,~can be extended to give the result of Theorem \ref{thm-ckflpp1}, as local search has not been very successful in dealing with the capacity and the cardinality constraints together. In particular, we present the following result for CFLPP:



\begin{theorem}
\label{thm-cflpp}
There is a polynomial time algorithm that approximates 
uniform capacitated facility location problem with penalties within a constant factor $(O(1/\epsilon))$ violating the capacities by a factor of at most $(1 + \epsilon)$, for $\epsilon>0$.
\end{theorem}



We next consider a more general problem of {\em Capacitated k-facility location problem with penalties} (C$k$FLPP) where-in we are given an additional bound $k$ on the maximum number of facilities that can be opened with the cost function same as that of CFLPP. We present first constant factor approximations for the problem with and without violating the cardinality, violating the capacities a little. Our results for C$k$FLPP also provide same results for the well known {\em capacitated $k$-median problem with penalties} (C$k$MP) as a special case. 
%
%
%
Though, some results are known for {\em uncapacitated $k$-median problem with penalties} ($k$MP)~\cite{charikar2001algorithms,Hajiaghayi-kMP,Wu-kMP} and {\em uncapacitated $k$-facility location problem with linear penalties}($k$FLPP)~\cite{wang-kFLPP} using primal dual schema and local search, no result has been obtained using LP-rounding to the best of our knowledge. Thus, our results for C$k$FLPP provide first LP-rounding based approximations for the uncapacitated variants of these problems as well. To the best of our knowledge, no results are known for the capacitated variants of the problems.


\begin{theorem}
\label{thm-ckflpp1}
There is a polynomial time algorithm that approximates uniform capacitated $k$-facility location problem with penalties within a constant factor $(O(1/\epsilon))$ violating the capacities by a factor of at most $(1 + \epsilon)$ and cardinality by a factor of at most $2$, for $\epsilon>0$.
\end{theorem}

\begin{theorem}
\label{thm-ckflpp2}
There is a polynomial time algorithm that approximates 
uniform capacitated $k$-facility location problem with penalties within a constant factor $(O(1/\epsilon^2))$ violating the capacities by a factor of at most $(2 + \epsilon)$, for $\epsilon>0$.
\end{theorem}

The integrality gaps of capacitated $k$-median(C$k$M)~\cite{Charikar:1999} apply to C$k$FLP as well, i.e., no constant factor approximation can be obtained by rounding a solution to standard LP, violating one of the cardinality/capacity by less than a factor of $2$ without violating the other. Thus the  violations in Theorems~\ref{thm-ckflpp1} and~\ref{thm-ckflpp2} are inevitable.

\textbf{Related Work:} Charikar \etal~\cite{charikar2001algorithms} were the first to define the problems of $k$-center, facility location and $k$-median with penalties/outliers and they showed that the 
Uncapacitated
FLPO 
has unbounded integrality gap with standard LP.
Theye get around the gap by guessing the most expensive facility opened by the optimal to give a $(3+\epsilon)$-factor approximation using primal dual technique. 
Friggstad \etal~\cite{FriggstadKR19}
gave a PTAS using multiswap local search on a restricted variant of the problem with uniform facility opening costs and doubling metrics.

For uncapacitated FLPP, a $3$-factor approximation using primal dual technique was given by Charikar \etal~\cite{charikar2001algorithms} which was subsequently improved to $2$ independently by Jain \etal~\cite{jain2003greedy} and Wang \etal~\cite{Wang2015penalties} using dual-fitting and a combination of primal-dual and greedy approach respectively. Later, Xu and Xu~\cite{xu2005lp} gave a $(2+2/e)$ approximation using LP rounding. Same authors~\cite{xu2009improved} improved the factor to $1.8526$
using a combination of primal-dual schema and local search. For linear penalties, using LP-rounding, Li \etal~\cite{LiDXX15} gave a $1.5148$-factor which was subsequently improved to $1.488$ by Byrka and Lewandowski~\cite{CFLPP_Byrka}.


Uncapacitated $k$-median problem with outlier($k$MO) is also known to have an unbounded integrality gap with standard LP~\cite{charikar2001algorithms}. Charikar \etal~\cite{charikar2001algorithms} gave a $4(1+1/\epsilon)$-approximate solution using primal-dual technique with $(1+\epsilon)$-factor violation in outliers. 
Friggstad~\etal~\cite{FriggstadKR19} used local search techniques to obtain $(3+\epsilon)$ and $(1+\epsilon)$-approximations with ($1+\epsilon$) violation in cardinality for general and doubling metric respectively.
The first true constant factor approximation  was given by Chen~\cite{chen-kMO}  using a combination of primal-dual and local search. 
The current best known result for the problem is due to 
Krishnaswamy \etal~\cite{krishnaswamy-kMO} who gave a $(7.081+\epsilon)$-factor approximation using iterative rounding and strengthened LP.

For $k$MP, Charikar \etal~\cite{charikar2001algorithms} gave a $4$-factor approximation using primal-dual schema. This was later improved to $(3+\epsilon)$ by Hajiaghayi \etal~\cite{Hajiaghayi-kMP} using local search technique. For uniform penalties, Wu \etal~\cite{Wu-kMP} gave a $(1+\sqrt{3}+\epsilon)$-approximation algorithm via pseudo-approximation.
The only result known for $k$FLPP is due to 
Wang \etal~\cite{wang-kFLPP} who used local search to obtain $(2+\sqrt{3}+\epsilon)$-approximation.

For CFLP, Shmoys \etal~\cite{Shmoys} gave the first constant factor($7$) approximation when the capacities are uniform, with a capacity blow-up of $7/2$, by rounding the solution to standard LP. Grover \etal \cite{GroverGKP18} 
reduced the capacity violation to $(1 + \epsilon)$, thereby showing that the capacity violation can be reduced to arbitrarily close to $0$ by roudning a solution to standard LP. An \etal~\cite{Anfocs2014} gave the first true constant factor(288) approximation for the problem (non-uniform), by strengthening the standard LP. 
Several results~\cite{KPR,ChudakW99,mathp,paltree,mahdian_universal,zhangchenye,Bansal} have been successfully obtained using local search with the current best being $5$-factor for non-uniform~\cite{Bansal} and $3$-factor for uniform capacities~\cite{Aggarwal}.

For C$k$M, constant factor approximations~\cite{Charikar,Charikar:1999,capkmshanfeili2014,ChuzhoyR05,capkmByrkaFRS2013,aardal2013,GroverGKP18} are known, that violate capacities or cardinality by a factor of $2$ or more. Violations were reduced to $(1 + \epsilon)$ by using strengthened LP~\cite{ByrkaRybicki2015,capkmshili2014,Lisoda2016,Demirci2016}. 
For uniform C$k$FLP, a constant factor approximation, with $(2 + \epsilon)$ violation in capacities, was given by Byrka \etal~\cite{capkmByrkaFRS2013} using dependent rounding. This was followed by two constant factor approximations by Grover \etal~\cite{GroverGKP18}, one with $(1 + \epsilon)$ violation in capacities and $2$-factor violation in cardinality and the other with $(2+\epsilon)$ factor violation in capacities only.


The only result known for the outlier/penalty variant of the problems in the capacitated setting
is  by Gupta \etal~\cite{GuptaG14}. They gave $(5.83+\epsilon)$ and $(8.532+\epsilon)$ factor approximation for CFLPP for uniform  and non-uniform capacities respectively using local search. A $25$-factor approximation for Capacitated $k$-center with outliers is given by Cygan and Kociumaka~\cite{Cygan-kCO}. To the best of our knowledge, no result is known for CFLPO and C$k$FLPP. 

\textbf{Our Techniques: }
Most of the work dealing with outliers/penalty uses primal-dual technique or a combination of the primal-dual/dual-fitting with greedy/local search schema. Since primal-dual technique has not been able to handle capacities, these works could not be extended to the capacitated variant of the problems. Local search alone does not perform well for outliers even for the uncapacitated variants. We use LP-rounding to obtain our desired claims.

The proposed framework works in two/three steps.  In the first step we identify the set of clients that serve as outliers (/pay penalty)  in our solution. 
 In the second step, we call upon the solution to the underlying problem without outliers(/penalty). In some cases, we are able to directly plug-in the solution of the underlying problem eliminating the need for the third step. However, in some other cases, we modify the solution of the underlying problem to open the facilities integrally. Once we identify the set of facilities to open, in third step, we can solve the transportation problem with outliers (/penalty) to obtain the integral assignments. Thus, in the rest of the paper, we only focus on identifying the set of facilities to open. Note that the LP for the transportation problem with outliers (/penalty) is TUM(totally unimodular) and hence provides an integral optimal solution.
 
 In step $2$ of CFLPP (and C$k$FLPP with cardinality violation), we raise the assignment of the remaining clients to $1$ and the openings accordingly so that the solution so obtained is a feasible solution for the LPs of CFLP (/C$k$FLP). In this case, we are able to directly plug-in the solution of CFLP (/C$k$FLP) eliminating the need for the third step. 
For CFLPO (and C$k$FLPP without cardinality violation), in our reduced problem after step $1$, we cannot simply plug-in the solutions of CFLP (/C$k$FLP) as raising the assignments of the remaining clients to $1$ leads to large violation in outliers (/violation in cardinality). Hence, in step $2$, we modify the solution of the underlying problem to open the facilities integrally preserving the extent to which clients are serviced from step $1$. 

\textbf{Organisation of the paper:}
In Section \ref{CFLPO}, we present our algorithm for CFLPO followed by
CFLPP 
in Section \ref{CFLPP}. Results for C$k$FLPP are presented in Section \ref{CkFLPP-(2+e)}. We conclude with some additional results based on the use of strengthened LP for the underlying problems in Section \ref{add_results}
and future scope in Section \ref{conclusion}.
 \section{Capacitated Facility Location Problem with Outliers}
 \label{CFLPO}


CFLPO can be formulated as the following integer program (IP):

\label{{unif-CFLP}}
$\text{Min} ~\mathcal{C}ostCFLPO(x,y,z) = \sum_{j \in \cliset}\sum_{i \in \facilityset}\dist{i}{j}x_{ij} + \sum_{i \in \facilityset}f_iy_i $
\begin{eqnarray}
\text{subject to} &\sum_{i \in \facilityset}{} x_{ij} + z_j \geq 1 & \forall ~\singleclient \in \clientset \label{LPFLP_const1}\\ 
&\sum_{j \in \cliset}{} x_{ij} \leq \capacity ~ y_i & \forall~ \singlefacility \in \facilityset \label{LPFLP_const2}\\ 
& x_{ij} \leq y_i & \forall~ \singlefacility \in \facilityset , ~\singleclient \in \clientset \label{LPFLP_const3}\\   
& \sum_{j \in \clientset}z_j \leq t &  \label{LPFLP_const4}\\   
& z_j,y_i,x_{ij} \in \left\lbrace 0,1 \right\rbrace  \label{LPFLP_const5}
\end{eqnarray}
where variable $y_i$ denotes whether facility $i$ is open or not, $z_j$ indicates whether client $j$ is an outlier and, $x_{ij}$ indicates if client $j$ is served by facility $i$ or not. 
Constraints \ref{LPFLP_const1} and \ref{LPFLP_const3} ensure that clients are either outliers or are assigned to opened facilities. Constraints \ref{LPFLP_const2} and Constraint \ref{LPFLP_const4} take care of the bounds on capacities and outliers resp.
LP-Relaxation of the problem is obtained by allowing the variables $ z_j,y_i, x_{ij} \in [0, 1]$. Call it $\lp{CFLPO}$. 

\textbf{Step $1$: Identifying the outliers: }We first identify the set of clients that will be treated as outliers in our solution. 
Let $\rho^{*} = <x^*, y^*, z^*>$ denote the optimal solution of $\lp{CFLPO}$ and $\opt{}$ denote the cost of $\rho^*$. 
For a given $0 < \epsilon < 1/2$, let $j \in \clientset$ be such that $z^*_j \geq (1-\epsilon)$, we treat such clients as outliers in our solution.
Let $\cliout$ be the set of these clients and $\clirem$ be the set of remaining clients.
Let $\hat{\rho} = <\hat{x}, \hat{y}, \hat{z}>$ be the solution so obtained. 
Facility openings and  the assignments of the remaining clients remain the same. Note that $\sum_{j \in \cliset}\hat{z} _j  \leq (\frac{1}{1-\epsilon})\sum_{j \in \cliset}z^*_j \leq (1+2\epsilon)t$ for $\epsilon \leq 1/2$. Also, $CostCFLPO(\hat{x},\hat{y},\hat{z}) \leq CostCFLPO(x^*,y^*,z^*)$.
    

\begin{figure}[]
	\begin{tabular}{ccc}
		\includegraphics[width=35mm,scale=0.5]{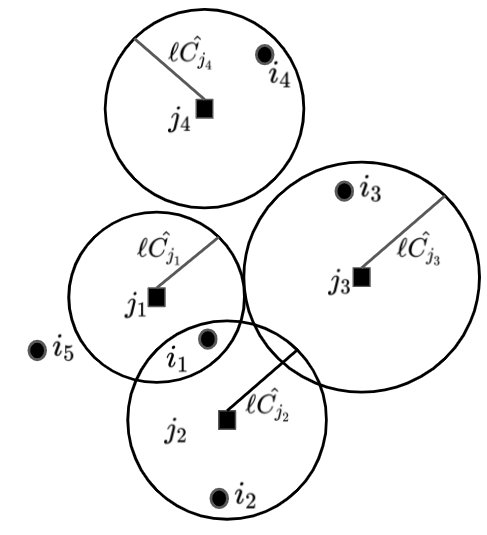}
		& 
	     \hspace{0.1cm} \includegraphics[width=35mm,scale=0.5]{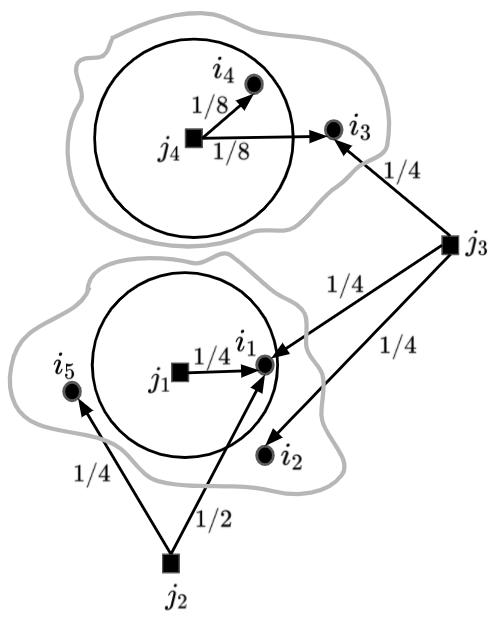}
		&
		\includegraphics[width=35mm,scale=0.5]{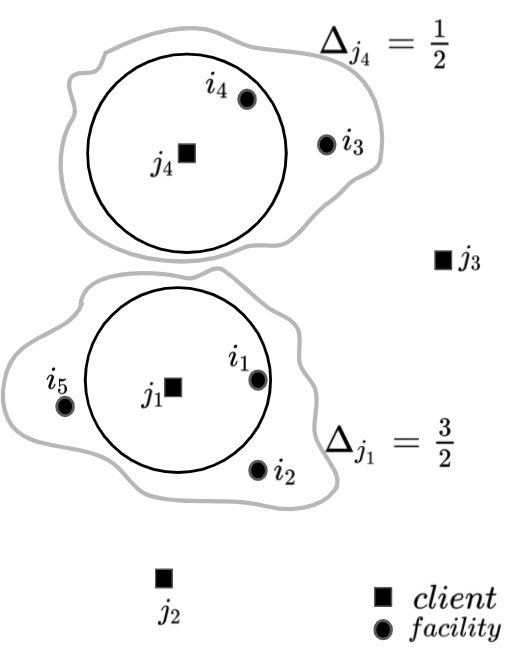}
		\\ 
		(a) & \ \  (b) & \ \ (c)
		\\
		\end{tabular}
		
		\caption{(a) Balls around the clients in $\clirem$.
		(b) Let $\ell=4$. Reduced set of clients $\clientset' = \{j_1,j_4\}$, partition of $\facilityset: \bundle{j_1} = \{i_1,i_2,i_5\},~\bundle{j_4} = \{i_3,i_4\}$ and assignment by LP solution.
		(c) Partitioning of demand: $\demandofj{j_1} = \sum_{j \in \clirem}(\hat{x}_{i_1j}+\hat{x}_{i_2j} + \hat{x}_{i_5j})$, $\demandofj{j_4} = \sum_{j \in \clirem}(\hat{x}_{i_3j}+\hat{x}_{i_4j})$.
	}\label{fig-clustering}\end{figure}

\textbf{Step $2$: Integrally Open Solution: }
\label{clustering-cflpo}
A solution $<x,y,z>$ is said to be an {\em integrally open solution} if each facility is either fully opened or fully closed, that is, $\forall i \in \facilityset$, $y_i$ is $0$ or $1$. We obtain an integrally open solution for the problem instance with the reduced  set $\clirem$ of clients. Recall that the clients in our reduced instance 
 need not be
fully served;
however, $ \sum_{i \in \facilityset} \hat{x}_{ij} > \epsilon,\forall j \in \clirem$
.We preserve the extent to which the clients in $\clirem$ are served while opening the facilities integrally. 

Let $\ell \geq 2$ be a tuneable parameter. We first sparsify the problem instance by removing some clients from the client set $\clirem$. 
This is done using standard clustering technique (refer to Figure \ref{fig-clustering}): for $j \in \clirem$, let $ \C{j} $ denote the average connection cost of $\singleclient$ in $\hat{\rho}$, i.e., ~$\C{j} = (\sum_{i \in \facilityset}{} \hat{x}_{ij} \dist{i}{j})/({\sum_{i \in \facilityset}{} \hat{x}_{ij}})$. Further, let $\ballofj{j}$ be the set of facilities within a distance $ \ell\C{j}$ of $j,~ i.e., \ \ballofj{j}  = \{ \singlefacility \in \facilityset \colon \dist{i}{j} \leq \ell \C{j} \} $. Then, total extent up to which facilities in $\ballofj{j}$ are opened under $\hat{\rho}$ is $\ge (1-1/\ell)(\sum_{i \in \facilityset} \hat{x}_{ij}) \geq \epsilon/2$. 
 Clients in $\clirem$ are considered in the non-decreasing order of 
$\ell \C{j}$. For a client $j$ at hand, remove all the clients $k : \dist{j}{k} \le 2 \ell \max \{\C{j}, \C{k}\}$ 
and repeat the process with the remaining clients. Let $\clientset'$ be the set of  clients remaining after all the clients in $\clirem$ have been considered. Clusters, of facilities, are formed around the  clients in $ \clientset'$ by assigning a facility to the cluster of the nearest client in $\clientset'$, i.e. if, for $j' \in \clientset'$, $\bundle{j'}$ denotes the cluster centered at $j'$ then a facility $i$ belongs to $\bundle{j'}$ if and only if $j'$ is the closest client in $\clientset'$ to $i$. Note that $\ballofj{j'} \subseteq \bundle{j'}$.
The clients in $\clientset'$ are called the cluster centers.
Any two cluster centers $j', k'$ in $\clientset'$ satisfy: $\dist{j'}{k'} > 2\ell~max\{ \C{j'}, \C{k'}\}$. 
For $j\in \clirem$, $j' \in \clientset'$, let $\loadjinc{j}{j'}$ be the extent up to which $j$ is served by the facilities in $\bundle{j'}$. Let $\Delta_{j'}= \sum_{j \in \clirem} \loadjinc{j}{j'}$, $\forall j' \in \cliset'$. A cluster is said to be $small$ if $\Delta_{j'} \le  u$ otherwise it is called $big$. Let $\csparse = \{j' \in \clientset': \Delta_{j'} \le  u\}$ and $\cdense = \clientset' \setminus \csparse$.

For $j' \in \csparse$, we open a cheapest facility say $\iofj$, in $\ballofj{j'}$ at a loss of factor ($2/\epsilon$) in the facility opening cost and transfer all the assignments coming into the cluster onto it at a loss of $4(\ell+1)$ factor in the service cost.\footnote{For a client $j$, there must be a cluster center $k'$ such that $\dist{j}{k'} \le 2 \ell \C{j} $. Then 
$\dist{\iofj}{j} \le \dist{j'}{j} + \dist{\iofj}{j'} \le  \dist{j'}{j} + \ell \C{j'} \le 2\dist{j'}{j}$ (for a far off $j$) $ \le 2(\dist{j'}{i} + \dist{i}{j})$ (for $i\in \bundle{j'}$) $ \le 2(\dist{k'}{i} + \dist{i}{j})$ (since $i$ is nearest to $j'$ in $\clientset'$) $\le 2(\dist{k'}{j} + \dist{j}{i} + \dist{i}{j}) \le 4 (\ell \C{j} + \dist{i}{j}$).}
 Since $\Delta_{j'} \le  u$ there is no  violation in the capacity. To handle big clusters, for every $j' \in \cdense$, we solve the LP: Min $Cost_{CI}(w) = \sum_{\singlefacility \in \bundle{j'}} ( f_i +\capacity \dist{i}{j'} ) \zofi{i}$ s.t. $\capacity \sum_{\singlefacility \in \bundle{j'}} \zofi{i} \geq \demandofj{j'}$ and $\zofi{i} \in [0,1]$. It can be easily shown that $\zofi{i} = \sum_{ j \in \clirem}\hat{x}_{i j}/\capacity$ is a feasible solution with cost at most $ \sum_{i \in \bundle{j'}} [\facilitycost \hat{y}_i +  \sum_{ j \in \clirem} \hat{x}_{i j} (\dist{i}{ j} + 2\ell \C{j})$].
	An almost integral\footnote{a solution is called an {\em almost integral}  if it has at most one fractionally opened facility.} solution $\primezofi{{}}$ is obtained by arranging the fractionally opened facilities in $\zofi{{}}$ in non-decreasing order of $\facilitycost  + \dist{i}{j'} \capacity $ and greedily transferring the openings $\zofi{{}}$ without increasing the cost of the solution.
	We obtain an integrally open solution $\hat{w}$ as follows: if opening of the fractionally opened facility, if any, is $\leq \epsilon$, we close it 
	else we open it at  $(1/\epsilon)$-factor loss in (facility) cost. 
	$ \Delta_{j'}$ is distributed equally to the facilities opened in the cluster incurring a loss of at most $(1+\epsilon)$-factor in capacities and (service) cost.

 By choosing $\ell=2$, and summing over all small and big clusters, we obtain an integrally open solution that violates the capacities by a factor of $(1+\epsilon)$ and is of cost bounded by $O(1/\epsilon)$. 
\section{Capacitated Facility Location Problem with Penalties}
\label{CFLPP}


In this section, we will reduce CFLPP to an instance of CFLP and use existing results of CFLP~\cite{GroverGKP18,Anfocs2014} as a black box to get the desired result.
CFLPP can be formulated as the following integer program (IP):

\label{{unif-CFLP}}
$\text{Min}~\mathcal{C}ostCFLPP(x,y,z) = \sum_{j \in \cliset}\sum_{i \in \facilityset}\dist{i}{j}x_{ij}   + \sum_{i \in \facilityset}f_iy_i + \sum_{j \in \cliset} p_j z_j$ 
\begin{eqnarray} 
\text{subject to} &\sum_{i \in \facilityset}{} x_{ij} + z_j \geq 1 & \forall ~\singleclient \in \clientset \label{LPCFLPP_const1}\\ 
&\sum_{j \in \cliset}{} x_{ij} \leq \capacity ~ y_i & \forall~ \singlefacility \in \facilityset \label{LPCFLPP_const2}\\ 
& x_{ij} \leq y_i & \forall~ \singlefacility \in \facilityset , ~\singleclient \in \clientset \label{LPCFLPP_const3}\\   
& z_j,y_i,x_{ij} \in \left\lbrace 0,1 \right\rbrace  \label{LPCFLPP_const4}
\end{eqnarray}
where variable $z_j$ denotes if client $j$ pays penalty or not. Constraints \ref{LPCFLPP_const1} ensure that every client is either served or pays penalty. 
LP-Relaxation of the problem is obtained by allowing the variables $ z_j,y_i, x_{ij} \in [0, 1]$.  Call it $\lp{CFLPP}$. 


\textbf{Step $1$: Identifying the clients that pay the penalty: }We first identify the clients that will penalty in our solution; rest of the clients are serviced fully. Let $\rho^{*} = <x^*, y^*, z^*>$ denote the optimal $LP$ solution of $\lp{CFLPP}$ and $\opt{CFLPP}$ denote the cost of $\rho^*$. 
For a given $0 < \epsilon < 1/2$, let $j \in \clientset$ be such that $ z^*_j \geq \epsilon$, we pay penalty for such clients in our solution incurring at most $(1/\epsilon)$-factor loss in penalty costs. Let $\clipen$ be the set of these clients and $\clirem$ be the set of rest of the clients. For $j \in \clirem,~\forall i \in \facilityset : x^*_{ij}> 0 $, we raise the assignment of $j$ on $i$ proportionately so that $j$ is fully serviced. Opening of $i$ is raised so as to satisfy constraints~(\ref{LPCFLPP_const3}). Let $\hat{\rho} = <\hat{x}, \hat{y}, \hat{z}>$ be the solution obtained. 
That is, 
$\forall j \in \clirem$ and $\forall i \in \facilityset$, set $\hat{x}_{ij} = x^*_{ij}/\sum_{i \in \facilityset}x^*_{ij}$ and $\hat{y}_i  = min \{1, y^*_i \cdot max_{j \in \clirem : x^*_{ij} > 0 } \{\hat{x}_{ij}/x^*_{ij} \}\}$. 
   Note that $ x^*_{ij} \le $\footnote{Wlog we assume that ${\sum_{i \in \facilityset}x^*_{ij}} \le 1$.} $\hat{x}_{ij} \le (\frac{1}{1-\epsilon}) x^*_{ij}$
and $y^*_i \le \hat{y}_i \le (\frac{1}{1-\epsilon}) y^*_i.$ It is easy to see that $\hat{\rho}$ is a feasible solution to $\lp{CFLP}$ except for violating the capacities by a factor of $(\frac{1}{1-\epsilon}) \leq (1+2\epsilon)$ for $\epsilon \leq 1/2$. Also, cost of $\hat{\rho}$ is bounded by $(1/\epsilon)\opt{CFLPP}$.

\textbf{Step $2$: Reducing to CFLP: } Next, we solve an instance of capacitated facility location where-in the client set is reduced to $\clirem$ and capacities are scaled up by a factor of  ($1+2\epsilon$).
$<\hat{x},\hat{y}>$ provides a feasible solution for the LP of CFLP with cost at most
 $(\frac{1}{1-\epsilon})(\sum_{i \in \facilityset}f_iy^*_i + \sum_{j \in \clirem}\sum_{i \in \facilityset}\dist{i}{j}x^*_{ij} )$. Let $<\Bar{x},\Bar{y}>$ be an $\alpha$-approximate solution for CFLP with $\gamma$-factor violation in capacities. Then, $<\Bar{x},\Bar{y}, \hat{z}>$ is a solution to CFLPP
of cost within a constant factor $(max \{ \alpha(\frac{1}{1-\epsilon}) , \frac{1}{\epsilon} \})$ of $\opt{CFLPP}$ violating the capacities by a factor of $\gamma(1 +2\epsilon)$. Using the result of Grover \etal~\cite{GroverGKP18}, $\alpha = O(1/\epsilon), \gamma = (1 + \epsilon)$, we arrive at Theorem~\ref{thm-cflpp}. Note that we cannot get rid of the violations in capacities even on using the stronger result of An \etal~\cite{Anfocs2014} for CFLP that uses strengthened LP. 

\section{Capacitated $k$-Facility Location Problem with penalties}
\label{CkFLPP-(2+e)}
LP for C$k$FLPP is same as that for CFLPP with the additional constraint: $\sum_{i \in \facilityset}y_i \leq k$. The Constraint ensures that at most $k$ facilities are opened in a feasible solution.
Using $O(1/\epsilon)$ factor approximation for C$k$FLP with $(1 + \epsilon)$ violation in capacities and $2/(1 + \epsilon)$ factor loss in cardinality of Grover \etal
~\cite{GroverGKP18}\footnote{A careful analysis shows that the cardinality loss in~\cite{GroverGKP18} is actually $2/(1 + \epsilon)$}, the approach of Section $\ref{CFLPP}$ leads to Theorem~\ref{thm-ckflpp1}. We next present a result without violating the cardinality. The capacity violation increases slightly to $(2 + \epsilon)$ in the process. Let $\rho^{*} = <x^*, y^*, z^*>$ denote the optimal $\lp{CkFLPP}$ solution and $\opt{CkFLPP}$ denote the cost of solution $\rho^*$. Let $\ell \geq 4$ be a fixed parameter.

\textbf{Step $1$: Identifying the set of clients that pay the penalty: }For a given $\epsilon \leq 1/4$, we identify the set of clients that pay the penalty in our solution in the same manner as was done in Section \ref{CFLPP}. Let $\hat{\rho}$ be solution so obtained. 
Clients in $\clirem$ are then served to an extent $\geq (1- \epsilon)$, i.e., $\sum_{i \in \facilityset} \hat{x}_{ij} \geq (1-\epsilon)$. Unlike CFLPP,  raising the assignments of clients in $\clirem$ to $1$ leads to $(1+2\epsilon)$ factor loss in cardinality. Hence we can not directly plug in the solution of C$k$FLP~\cite{GroverGKP18} as the approach relies on the fact that clients are served to full extent to guarantee sufficient opening within a cluster.
 
\textbf{Step $2$: Obtaining an integrally open solution:} We modify the constant factor approximation  for C$k$FLP by Grover \etal~\cite{GroverGKP18} that violates the capacities by a factor of $(2 + \epsilon)$ without violating the cardinality to obtain an integrally open solution. The extent to which clients are served is preserved from Step $1$. 





Step $2$ works in two phases: In phase I, \ the problem instance is sparsified using clustering techniques
in the same manner as explained in Section \ref{clustering-cflpo}. $\hat{C}_j$, the average connection cost of a client $j \in \clirem$ in  $\hat{\rho}$, is
$(\sum_{i \in \facilityset}{} \hat{x}_{ij} \dist{i}{j})/\sum_{i \in \facilityset}{} \hat{x}_{ij}$. 
In this case, choosing $\ell \leq 1/\epsilon$, total extent up to which facilities in $\ballofj{j}$ are opened under $\hat{\rho}$ is $\ge (1-\frac{1}{\ell})(\sum_{i \in \facilityset} \hat{x}_{ij}) \geq (1-\frac{1}{\ell})^2$.
Small and big clusters are defined in the same manner as in Section \ref{clustering-cflpo}. 
We move the demand $\Delta_{j'}$ to $j'$ before moving to Phase II incurring a cost bounded by $2(\ell+1)CostCkFLPP(\hat{\rho})$. Refer to Appendix~\ref{cost_j'_to_center} for the proof.
    \begin{enumerate}
    \item For a big cluster, sufficient facilities are opened in it so that its entire demand can be assigned to the opened facilities within the cluster; we call such clusters as {\em self-sufficient}. 
    \item
     For a small cluster, we are not able to guarantee this. That is, we may not be able to open even one facility in a small cluster. We try to send the unmet demand of such a cluster to one or more nearby clusters in which we guarantee that its demand is assigned to the opened facilities within the claimed (capacity) bounds. To achieve the goal, a forest of routing trees on cluster centers is defined.
     The edges costs are non-increasing as we move up the tree.
     \item 
     Each tree in the routing forest is decomposed into a number of subtrees. The set of clusters in each subtree is called a {\em bundle}. Size of each bundle is chosen appropriately to make sure that we are able to open sufficient number of facilities within each bundle so that all except at most $u$ units of its demand is served within the \MC itself. We write an auxiliary LP (ALP) to identify the facilities to be opened integrally in these bundles. ALP is defined such that the opened facilities are well spread out amongst the clusters of a \MC (it is ensured that all but at most one small clusters have a facility opened in it and demand of a big cluster is satisfied within the cluster itself.)

     \item An iterative rounding algorithm is then used to obtain a {\em pseudo-integral} solution that has at most two fractionally opened facilities, for the ALP.
     
     \item Pseudo integral solution is converted into an integrally open solution by opening the facility with larger opening to full extent. 
\end{enumerate}
     
    
We start with describing the routing tree which is used to form the bundles. 

\subsection{Constructing the Routing Trees}

We define a directed graph $G=(V,E)$ on set $\clientset'$ of cluster centers. For  $j' \in \csparse$, let $\pi(j')$ be the nearest other cluster center of $j'$ in $\cliset'$ 
and for $j' \in \cdense,~\pi(j') = j'$. The graph $G$ consists of edges $(j',\pi(j'))$. Note that each connected component of the graph is a tree except possibly a 2-cycle at the root, we delete any one of the two edges from the cycle arbitrarily. The resulting graph is a forest.

Following Lemma will be helpful in providing a feasible solution of bounded cost for the auxiliary LP. For a client $j' \in \csparse$, Lemma \ref{lemma-costbd2} bounds the cost of serving major part (at least $ 1-1/\ell$ extent) of $\demandofj{j'}$.
Let a client $j'$ is served to an extent of $\gamma$ (in $\hat{\rho}$) within its cluster $\bundle{j'}$ and  to an extent of $\delta$ outside. 
Lemma \ref{lemma-costbd2} bounds  the cost of serving $\gamma$ extent of \demandofj{j'} within $\bundle{j'}$
and sending $(1 - 1/\ell) - \gamma$ extent (if positive) of it to the nearest cluster center.
Note that $\gamma + \delta \geq (1-1/\ell)$.


\begin{lemma}	
	\label{lemma-costbd2}				
	$\sum_{j' \in { \csparse}}{} \bard{j'} (\sum_{i \in {\bundle{j'}}}\dist{i}{j'} \hat{x}_{ij'} + \dist{j'}{\pi(j')} ( \frac{\ell-1}{\ell} - min \{ \frac{\ell-1}{\ell},\sum_{i \in \bundle{j'}}\hat{x}_{ij'} \\ \} ))  \leq 6 \cdot  CostCkFLPP(\hat{\rho})$. \end{lemma}
\begin{proof}
Refer to Appendix~\ref{prf_lemma_costbd2}.
\end{proof}


\begin{figure}[t]

	\begin{tabular}{cccc}
		\includegraphics[width=30mm,scale=0.5]{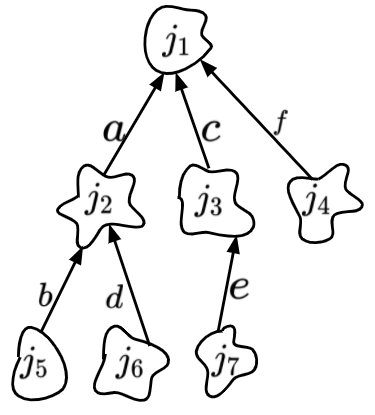}
		& 
		\includegraphics[width=60mm,scale=0.5]{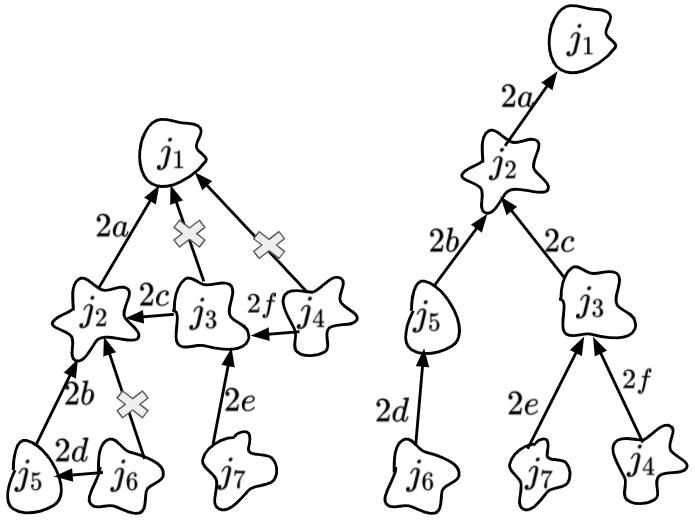}
		& 
		\includegraphics[width=30mm,scale=0.5]{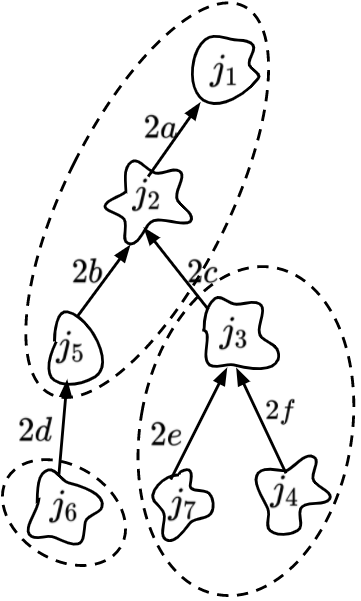} 	
		\\
		\textbf{(a)}& \ \
		\textbf{(b)}& \ \
		\ \ \textbf{(c)}
	\end{tabular}
	\caption{(a) A Tree $T$ of unbounded in-degree. $a<c<f$ , $b<d$.
		(b) A Binary Tree $T'$ where each node has in-degree at most $2$.
		(c) Formation of bundles for $\ell =6$. 
	}
	\label{bin-tree}
\end{figure}

For a $j' \in \csparse$, if we are unable to open a facility in $\bundle{j'}$ , we try to satisfy its demand from $\pi(j')$. However, as the in-degree of a node in the above routing trees is unbounded, this can lead to arbitrarily large capacity violations. Thus, we convert these trees into binary trees
using the standard procedure (refer to Figure \ref{bin-tree}). The binary routing trees  have the following properties:
$(i)$ there is at most one big cluster and if present, it must be the root of the tree, $(ii)$ the in-degree of root is at most $1$ and,
$(iii)$ the edge costs are non-increasing as we move up from the leaves to the root. 
Let $\sigmaone(j')$ be the parent of cluster $j'$ in the newly formed binary trees.
Then, $\dist{j'}{\sigmaone(j')} \leq 2 \dist{j'}{\pi(j')}$; hence we update the weights on all the edges to $2 \dist{j'}{\pi(j')}$ and we get the following lemma,
\begin{lemma}
$\sum_{j' \in {\csparse}}{} \bard{j'} \big( 
\sum_{i \in {\bundle{j'}}}{}
\dist{i}{j'} \hat{x}_{ij'} +
\dist{j'}{\sigmaone(j')} ( \frac{\ell-1}{\ell} - min \{ \frac{\ell-1}{\ell},\sum_{i \in \bundle{j'}}{} \\\hat{x}_{ij'} \} )
\big) \leq 12 \cdot CostCkFLPP(\hat{\rho})	$
\label{factor12}
\end{lemma}
\subsection{Integrally open solution: constructing the \MCs and writing the ALP}
\label{ALP1}

In this section, we write an auxiliary LP to identify the set of facilities to be opened integrally. We make bundles consisting of $q$ clusters where $q$ is chosen so that we can guarantee to open at least $q - 1$ facilities in each bundle.
Recall that each cluster has at least $(1 - 1/\ell)^2$ opening in it. As $q (1 - 1/\ell)^2 \ge q - 1$ for $q \le \ell/2$, 
we make \MCs consisting of $q = \ell/2$ clusters and open at least $q - 1$ facilities in it. Bundles are formed as follows:
for every binary routing tree $\mathcal{T}$, make bundles by processing nodes of the tree $\mathcal{T}$ greedily from the root node. Starting from the root node, extend the \MC by adding a node that is connected to the \MC by a cheapest edge. We grow the \MC until either we have included $\ell/2$ nodes(/clusters) or there are no more nodes to add (i.e., we have reached the leaves of the tree). Remove the nodes of the \MC and repeat the above process if there are more nodes in the tree (the tree could have gotten disconnected after removing the nodes included in the bundle). Each \MC is a binary rooted tree in itself.
The construction imposes a natural rooted tree structure on the \MCs also. Some \MCs towards the leaves may have less than $\ell/2$ clusters. Refer to Figure 2(c).

Let $\Bundler$ be a \MC (whose binary rooted tree is) rooted at cluster $r$. Let $\beta_r$ (0 or 1) be the number of big clusters  and $\sigma_r$ ($\leq \ell/2$) be the number of small clusters in $\Bundler$. Further, let $\mathcal{S}(\Bundler)$ be the sub graph of $\mathcal{T}$ induced by the nodes in $\Bundler$. We use cluster centers, clusters and nodes in the graph interchangeably.
We write an ALP so as to open $\alpha_r = \max\{0,\sigma_r-1\}$ facilities in the small clusters of $\Bundler$ (Constraint~\ref{LP-Meta-Clusters_const2-3}) and $\floor{\Delta_{j'}/u}$ facilities in a big cluster centered at $j'$, if any (Constraint~\ref{LP-meta-clusters_const1-dense-3}). In order to ensure that all but one small clusters have a facility opened in it
(Constraint~\ref{LP-meta-clusters_const1-sparse-3}) makes sure that not too many facilities are opened in one small cluster. 
ALP bounds the cost of sending the demand of such a cluster to its parent. However, it doesn't send the entire demand to its parent but does so for a major portion of it ($(1 - 1/\ell)$ extent of it).

\begin{figure}[t]
\begin{center}
   \includegraphics[width=60mm,scale=0.5]{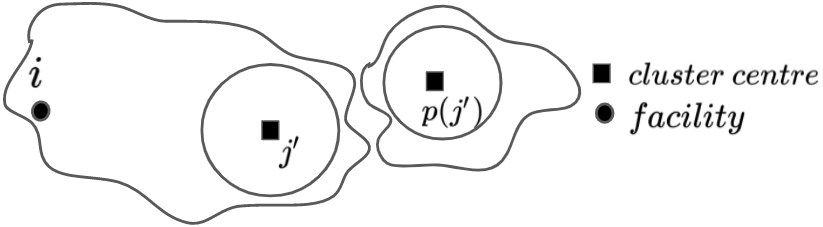} 
\end{center}
		
	\caption{There may be a facility $i$ in $\T{j'}$ that is farther from $j'$ than $\sigmaone(j')$. }
	\label{fj}
\end{figure}

For $j' \in \csparse$, we would like to open a facility in $\bundle{j'}$ only if it is no farther than $\sigmaone(j')$ from $j'$ (refer to Figure~\ref{fj}). Thus,
we define $\T{j'} = \{i \in \bundle{j'}:\dist{i}{j'} \leq \dist{j'}{\sigmaone(j')}\}$ if $j' \in \csparse$  and $\T{j'} = \neighbor{j'}$ if $j' \in \cdense$. 
Let $w_i$ denotes whether facility $i$ is opened in the solution or not.
Thus we arrive at the following ALP:

\noindent $\text{Min}~CostALP(w)= \sum_{j' \in \csparse }{} \demandofj{j'} [\sum_{i \in \bundle{j'}}\dist{i}{j'} w_{i} + \dist{j'}{\sigmaone(j') }  ( \frac{\ell-1}{\ell} - min \{ \frac{\ell-1}{\ell}, \sum_{i \in \bundle{j'}}{} w_{i} \})] \\ + \capacity \sum_{j' \in \crich } \sum_{\singlefacility \in \bundle{j'}} \dist{i}{j'} w_i + \sum_{i \in \facilityset} f_i w_i$ 
\begin{eqnarray}  
\text{s.t.} &\sum_{i \in \T{j'}}{} w_{i} \leq 1 &\forall ~j' \in \csparse  \label{LP-meta-clusters_const1-sparse-3}\\
&\sum_{i \in  \T{j'}} w_{i} \ge \floor{\Delta_{j'}/u}&\forall ~j' \in \cdense  \label{LP-meta-clusters_const1-dense-3}\\
&\sum_{j' \in \Bundler\cap \csparse}~\sum_{i \in \T{j'}} w_{i} \geq \alpha_r &\forall~r:\Bundler \text{ is a Bundle}  \label{LP-Meta-Clusters_const2-3}\\
&\sum_{i \in \facilityset}{} w_i \leq k
\label{LP-meta-clusters-const2-5}\\
&0 \leq w_i \leq 1 &\forall ~i \in \facilityset \label{LP-meta-clusters_const5-3}
\end{eqnarray}

It can be shown that optimal cost of the ALP is bounded.
For this, we will construct a feasible solution of bounded cost.
For all $i \in \facilityset$, let ~$\load{i} = \sum_{ j \in \clirem}{} \hat{x}_{i j}$. 
	For $j' \in \cdense,~i \in \T{j'}$, set $w'_i =  \frac{d_i}{u}$ 
	. For $j' \in \csparse$, set $w'_i = \hat{x}_{ij'}$ for $i \in \T{j'}$ and $w'_i = 0$ for $i \in \neighbor{j'} \setminus \T{j'}$. It can be shown that $w'$ is a feasible solution to the ALP and that its cost is bounded by $(2\ell+14) \cdot  CostCkFLPP(\hat{\rho})$.
	Refer to Appendix~\ref{feasiblesolution-ALP} for detailed proof.

A pseudo-integral solution $\tilde{w}$ 
of cost $CostALP(\tilde{w})$ bounded by $(2\ell+14) \cdot CostCkFLPP(\hat{\rho})$, is obtained by using an iterative rounding algorithm: compute an extreme point solution $w^{(o)}$ for the original ALP. LP for the next iteration is obtained by removing the integral variables and updating the constraints accordingly. 
%
%
The process is repeated until either all the variables are integral or all of them are fractional. 
If all the variables are integral we are done, otherwise
we use the properties of extreme point solution to claim that the number of non-zero fractional variables are at most two and that they must be in the same bundle. 
%
It is easy to show that the cost of the ALP is non-increasing over the iterations. 

To obtain an integrally open solution $\bar{w}$,
we open the facility with larger fractional opening and close the other incurring  a loss of $2$ factor  in the facility cost.
If the two facilities are in the same cluster,  we loose at most $2$ factor in the first/third term of the objective function.  Otherwise, let $j'$ be the cluster in which we closed down the facility opened by ALP; note that since $(\ell - 1)/\ell \ge 3/4$ for $\ell \ge 4$, the ALP bounds the cost of sending at least $1/4$ of the demand of $j'$ to $\sigmaone(j')$ in $\tilde{w}$.  
Thus, the cost of sending $(1 - 1/\ell)$ extent of demand of $j'$ to $\sigmaone(j')$ is bounded by
$4\cdot((\ell-1)/\ell)$.  Hence we loose a
factor of at most $4$ in the second term of the objective function in this case. Thus, $CostALP(\bar{w}) \leq 4 \cdot CostALP(\tilde{w}))$.


Having obtained an integrally opened solution, we first define the assignments within a bundle. A cluster with a facility opened in it is self-sufficient.
If no facility is opened in a small cluster centered at $j'$, its demand is assigned to $\sigmaone(j')$ if $j'$ is not the root cluster of the bundle at a loss of at most factor $2$ and $3$ in capacities for $\sigmaone(j')$ small and big respectively and at a
cost of $\Delta_{j'} \dist{j'}{\sigmaone(j')}$.

 

If $j'$ is the root cluster of its bundle, its demand is distributed amongst the facilities opened in the parent bundle.  Let $\delta_r =  \floordjbyu{\dense} + \max \{0, \sigma_r - 1\}$. The violation in the capacities is bounded by $\frac{Total\  demand \  to \  be \  served \  in \ a \  Bundle}{Total \ capacity \ in \ a \ Bundle} \le \frac{(\floor{\Delta_{j'}/u} + 1 + \sigma_r)u + (\frac{\ell}{2}+1)u }{\delta_ru} \leq \frac{(\delta_r+2)u + (\frac{\ell}{2}+1)u}{\delta_r u} \leq \frac{2(\delta_r+2)u}{\delta_r u} = 2+\frac{4}{\delta_r} \leq 2+ \frac{8}{(\ell-2)}$ where the last equality follows as $\delta_r \geq (\beta_r+ \sigma_r -1) = (\ell/2 -1)$ for a non-leaf bundle.
It is easy to see that the edge costs in the parent bundle are all less than $\dist{j'}{\sigmaone(j')}$.
By triangle inequality, the per unit cost of this assignment to a facility $i$ in a cluster centered at $k'$ in the parent bundle is bounded by $\dist{j'}{k'} + \dist{k'}{i}$. As there are at most $\ell/2$ edges in a bundle,   $\dist{j'}{k'} \le O(\ell)\dist{j'}{\sigmaone(j')}$. Also, $\dist{k'}{i} \le \dist{k'}{\sigmaone(k')}$ which is further $\le \dist{j'}{\sigmaone(j')} $. Refer to Figure~\ref{cost_bd_fig}.

Thus, if no facility is opened in a small cluster centered at $j'$, its demand can be assigned within a distance of $O(\ell)\dist{j'}{\sigmaone(j')} \le O(\ell)\dist{j'}{\sigmaone(j')} (\ell/(\ell -1)) (\frac{\ell-1}{\ell} - min \{ \frac{\ell-1}{\ell},\sum_{i \in \bundle{j'}}\bar{w}_{i}\})$.

\begin{figure}[]
\begin{center}
   \includegraphics[width=45mm,scale=0.5]{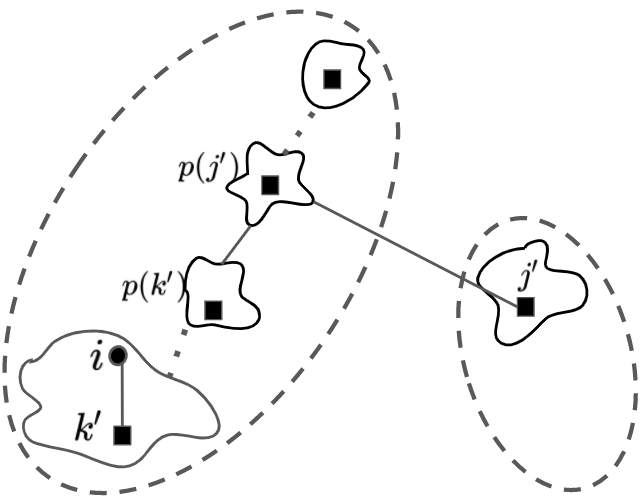} 
\end{center}
	\caption{$c(i,k') \leq c(k',p(k')) \leq c(j',p(j'))$}
	\label{cost_bd_fig}
\end{figure}

%
%
\subsection{($2+\epsilon$)-factor loss in capacities}
 Note that the only case in Section \ref{ALP1} where the capacities are violated by a factor of $3$ was when
the unmet demand of a small cluster is assigned to a big cluster and the big cluster already has a demand slightly less than $2 \capacity$. Consider a $\MC$ $\Bundler$ with a big cluster $j_b$ and $j_s$ as the only small child cluster.
In the above scenario, we consider $j_s$ with $j_b$ instead of considering it with remaining small clusters in $\Bundler$ and will open one more facility in either $\T{j_s}$ or $\T{j_b}$. To do so, we partition the $\MC$ $\Bundler$ into 
$\bundleone$ and $\bundletwo$ as follows: let $\alpha_{r}^b$ and $\alpha_{r}^s$ be the number of facilities we are going to open in $\bundleone$ and $\bundletwo$ respectively. Let  $\sigma_{r}^s$ be the number of small clusters in $\bundletwo$.

\begin{enumerate}
    \item If $\Bundler \cap \csparse = \phi$ or $(\frac{\Delta_{j_b}}{\capacity} -\floor{\frac{\Delta_{j_b}}{\capacity}}) \leq \delta$, then set $\bundleone=\Bundler \cap \cdense$, $\bundletwo=\Bundler \cap \csparse$, $\sigma_{r}^s=\sigma_r$ and $\alpha_{r}^b= \floor{\frac{\Delta_{j_b}}{u}}$.
    
    \item Else, $\frac{\Delta_{j_b}}{\capacity} -\floor{\frac{\Delta_{j_b}}{\capacity}} > \delta$, then set set $\bundleone=(\Bundler \cap \cdense)\cup \{j_s\}$, $\bundletwo=(\Bundler \cap \csparse) \setminus \{j_s\}$, $\sigma_{r}^s=max \{ \sigma_r-1, 0\}$ and $\alpha_{r}^b=\floor{\frac{\Delta_{j_b}}{u}} + 1$.
\end{enumerate}

Next we modify our ALP so as to open at least $\alpha_{r}^b$ and $\alpha_{r}^s = max \{ 0, \sigma_{r}^s-1 \}$ facilities in  $\bundleone$ and $\bundletwo$ respectively. 
We replace Constraints~\ref{LP-meta-clusters_const1-dense-3} and~\ref{LP-Meta-Clusters_const2-3} in the ALP of Section~\ref{ALP1} by the following constraints

\begin{eqnarray} 
&\sum_{j' \in \bundleone}~\sum_{i \in \T{j'}} w_{i} \geq \alpha_{r}^b &\forall~r:\bundleone \label{ALPnew-1}\\
&\sum_{j' \in \bundletwo}~\sum_{i \in \T{j'}} w_{i} \geq \alpha_r^s &\forall~r:\bundletwo  \label{ALPnew-2}
\end{eqnarray}

Now it is possible that two small clusters have no facility opened in them, one in $\bundleone$ and the other in $\bundletwo$. Let $j'$ be a child of $j_s$ and no facility is opened in both $j'$ and $j_s$. Note that in this case we must have opened $\alpha_{r}^b=\floor{\frac{\Delta_{j_b}}{u}} + 1$ facilities in $\T{j_b}$. Distribute the demands of $j'$ and $j_s$ to these opened facilities. The capacity violation is bounded by 
$\frac{\Delta_{j_b} + 2\capacity}{(\floor{\frac{\Delta_{j_b}}{u}} + 1)\capacity} \le 2$.

Since the total opening in a bundle remains same as that in Section~\ref{ALP1} demand from the children bundles are accommodated in the same capacity bounds.


\section{Additional Results using strengthened LPs}
\label{add_results}
%
%
For CFLPP, similar results can be obtained for non-uniform capacities also using the result of An \etal~\cite{Anfocs2014} as black box in step $2$. Note that, our results for C$k$FLPP also provide same results for the well known {\em capacitated $k$-median problem with penalties} (C$k$MP) as a special case. For C$k$MP, results of  Li~\cite{capkmshili2014} or Bryka \etal~\cite{ByrkaRybicki2015} can be used as black box in step $2$
of Section \ref{CFLPP} 
to reduce the violation in cardinality to $(1+\epsilon)$ for uniform capacities. Similar results can be obtained for non-uniform capacities  for C$k$MP using the result of Li~\cite{Lisoda2016} or Demirci \etal ~\cite{Demirci2016} in step $2$.

\section{Conclusion and Future Work}
\label{conclusion}
In this paper, we presented
a framework for designing approximation algorithms for capacitated facility location problems with outliers/penalty and applied it to obtain first constant factor approximations for some very fundamental problems like CFLPO and C$k$FLPP. 
Our solutions incur a slight violations in capacities, ($1 + \epsilon$) for the problems without cardinality constraint and ($2 + \epsilon$) for the problems with the cardinality constraint. For the outlier variant, we also incur a small loss ($1 + \epsilon$) in outliers. Due to the unbounded integrality gaps in the standard LP of the underlying problems, the violations are inevitable. Thus we achieve the best possible by rounding the solution of standard LP for these problems. The results of C$k$FLPP should be extendable to another closely related problem, Capacitated Knapsack Median with Penalties, by guessing the most expensive facility opened by the optimal solution and enumeration techniques.

It would be interesting to obtain similar results for the outlier variants of Capacitated $k$-Facility Location  and/or Capacitated $k$-Median. Three hard bounds viz. capacities, cardinality and the outliers, make the problems very challenging. Due to the integrality gap in LP of C$k$M, we must violate at least one of cardinality and capacities by a factor of at least $2$ and by the gap in LP of $k$MO, we must violate at least one of cardinality and outlier constraint. The challenge is to keep the violations low (some small constant). In particular, it would be interesting to see a result for C$k$MO violating only cardinality and/or a result violating any two constraints up to small constants or violating all three by $(1 + \epsilon)$. Also, it will be interesting and challenging to get constant approximations for the problems(CFLPO/C$k$MP/C$k$FLPP) without violating capacities/cardinality using strengthened LPs or local search.






\bibliographystyle{plainurl}
\bibliography{ref_master}
\appendix
\section{Appendix}
\subsection{Cost bound for moving demand $\Delta_{j'}$ to $j'$}
\label{cost_j'_to_center}
	Let $j' \in \clientset'$, 
	$ j \in \clirem$ and $i \in \bundle{j'}$, using the fact that $\dist{j'}{ j} \le 2\dist{i}{ j} + 2 \ell\C{ j}$ we get,
	
	$\sumlimits{j' \in \cliset'}{}\dist{j'}{ j} \sum_{i \in \bundle{j'}} \hat{x}_{ij} \Delta_j \leq 2\sumlimits{j' \in \cliset'}{} \sum_{i \in \bundle{j'}}{}(\dist{i}{j}+\ell\C{j})\hat{x}_{ij} \Delta_j$
	
	\hspace{1.6in}$= 2\sumlimits{j' \in \cliset'}{} \sum_{i \in \bundle{j'}}{}\dist{i}{j}\hat{x}_{ij} \Delta_j+ 2\ell\C{j}\sum_{j' \in \cliset'}{} \sumlimits{i \in \bundle{j'}}{}\hat{x}_{ij} \Delta_j$
	
	\hspace{1.6in}$= 2 \sumlimits{i \in \facilityset}{} \dist{i}{j}\hat{x}_{ij} \Delta_j+ 2\ell\C{j}\sumlimits{i \in \facilityset}{} \hat{x}_{ij} \Delta_j$

	Substituting the value of $\C{j}$, we get
	
	$\sumlimits{j' \in \cliset'}{}\dist{j'}{ j} \sum_{i \in \bundle{j'}} \hat{x}_{ij} \Delta_j \leq 2\sumlimits{i \in \facilityset}{}\dist{i}{j}\hat{x}_{ij} \Delta_j + 2\ell(\frac{\sum_{i \in \facilityset}{} \hat{x}_{ij} \dist{i}{j}}{\sum_{i \in \facilityset}{} \hat{x}_{ij}}) \sum_{i \in \facilityset}{} \hat{x}_{ij} \Delta_j$

	\hspace{1.2in}  $= 2(\ell+1)\sumlimits{i \in \facilityset}{}\dist{i}{j}\hat{x}_{ij} \Delta_j$. 
	
	Summing over all $j \in \clirem$, we get the desired claim.

\subsection{Proof of Lemma~\ref{lemma-costbd2}}
\label{prf_lemma_costbd2}
    The second term of LHS is
	
	$\leq \sum_{j' \in { \csparse}}{} \bard{j'} \dist{j'}{\pi(j')} ( \sum_{i \in \facilityset} \hat{x}_{ij'} -\sum_{i \in \bundle{j'}}\hat{x}_{ij'} ) \  $ \footnote{$\frac{\ell-1}{\ell} - min \{ \frac{\ell-1}{\ell},\sum_{i \in \bundle{j'}}\hat{x}_{ij'} \} = 0$ if $min \{ \frac{\ell-1}{\ell},\sum_{i \in \bundle{j'}}\hat{x}_{ij'} \} = (\frac{\ell-1}{\ell})$ and it is $\leq ( \sum_{i \in \facilityset} \hat{x}_{ij'} -\sum_{i \in \bundle{j'}}\hat{x}_{ij'} )$ otherwise, 
	where the inequality follows as $\sum_{i \in \facilityset} \hat{x}_{ij'} \ge (1- 1/\ell)$.}
	
	$= \sum_{j' \in { \csparse}}{} \bard{j'} \big( \sum_{i \notin {\bundle{j'}}}{}
	\dist{j'}{\pi(j')} \hat{x}_{ij'}\big) $
	
	$\leq\sum_{j' \in { \csparse}}{} \bard{j'}\big( \sum_{ k' \in { \cliset'}: k'\neq j'~}{} \sum_{i \in \bundle{ k'}}{}\dist{j'}{k'}{} \hat{x}_{ij'}\big)$
	
	 $=\sum_{j' \in { \csparse}}{} \bard{j'}\big( \sum_{ k' \in { \cliset'}: k'\neq j'~}{} \sum_{i \in \bundle{ k'}}{}(\dist{i}{j'} + \dist{i}{k'})\hat{x}_{ij'}\big)$ 
	(by triangle inequality)
	
	$ \leq \sum_{j' \in { \csparse}}{} \bard{j'} \big(\sum_{k'\in { \cliset'}: k'\neq j'~}{} \sum_{i \in \bundle{ k'}}{} 2 \dist{i}{j'} \hat{x}_{ij'} \big)$ 	
	(as $i \in \bundle{k'}$ and not $\bundle{j'}$.)
	
	Adding the first term of the claim and applying the following Lemma, we get the desired claim.

	\begin{lemma}~\cite{GroverGKP18}
	\label{lemma-costbd1}
	$\sum_{j' \in {\cliset'}}{} \Delta_{j'}\sum_{i \in {\facset}}{}\dist{i}{j'}\hat{x}_{ij'}\leq 3\sum_{j\in {\clirem}}\sum_{i \in {\facset}}{}\dist{i}{j}\hat{x}_{ij} = \\ 3CostCkFLPP(\hat{\rho})$
\end{lemma}

\begin{proof}
	By definition of $\Delta_{j'}$ and $\C{j'}$ we have, 
	
	$\sum{j' \in {\cliset'}}{} \Delta_{j'} 
	\sum_{i \in {\facset}}{}
	\dist{i}{j'}\hat{x}_{ij'} = \sum{j' \in {\cliset'}}{}
	\big( \sum{ j \in \clirem}{}
	\sum_{i  \in \bundle{j'}}\hat{x}_{ij} \Delta_j\big) \C{j'}$ 
	
	$= \sum{j' \in {\cliset'}}{}
	\big( \sumlimits{ j \in \clirem : \dist{j'}{j} \leq \ell\C{j'}}{}
	\sum_{i  \in \bundle{j'}}\hat{x}_{ij} \Delta_j \C{j'}+ \sumlimits{ j  \in \clirem : \dist{j'}{j} > \ell\C{j'}}{}
	\sum_{i  \in \bundle{j'}}\hat{x}_{ij} \Delta_j \C{j'}\big)$
	
	Consider the first term in the sum on RHS. Using  the fact that for $ j \in \clientset \setminus \clientset'$, if $\dist{j'}{ j} \leq \ell\C{j'}$, then $\ell\C{j'} \le 2\ell\C{j}$., we get
	
	$\sumlimits{j' \in {\cliset'}}{}
	\sum_{ j \in \clirem : \dist{j'}{j} \leq \ell\C{j'}}{}
	\sum_{i  \in \bundle{j'}}\hat{x}_{ij} \Delta_j \C{j'}
	\leq 2\sumlimits{j' \in {\cliset'}}{}
	\sum_{j  \in \clirem : \dist{j'}{j}\leq \ell\C{j'}}{}
	\sum_{i  \in \bundle{j'}}\hat{x}_{ij}\Delta_j \C{j}$
	
Next, considering the second term in sum on RHS.
	
	$\sumap{j' \in {\cliset'}}\sum_{j \in \clirem:\dist{j'}{j} > \ell\C{j'}} \sum_{i  \in \bundle{j'}}\hat{x}_{ij} \Delta_j \C{j'}$ 
	
	 $< \frac{1}{\ell} \sumlimits{j' \in {\cliset'}}{}\sumlimits{j  \in \clirem:\dist{j'}{j} > \ell\C{j'}}{}\sumlimits{i \in \bundle{j'}}{} \hat{x}_{ij} \Delta_j \dist{j'}{j}$
	
	 $\leq\frac{1}{\ell} \sumlimits{j' \in {\cliset'}}{} \sum_{j \in \clirem: \dist{j'}{j} > \ell\C{j'}}{} \sumlimits{i \in \bundle{j'}}{}
	\hat{x}_{ij} \Delta_j (2\dist{i}{j}+ 2\ell \C{j})$ as $\dist{j'}{ j} \le 2\dist{i}{ j} + 2 \ell\C{ j}$.
	
	$\leq \sumlimits{j' \in {\cliset'}}{}
	\sumlimits{j \in \clirem: \dist{j'}{j} > \ell\C{j'}}{} \sumlimits{i \in \bundle{j'}}{}\hat{x}_{ij} \Delta_j (\dist{i}{ j} + 2\C{j})$ as $\frac{2}{\ell}\leq 1$. 
	
	Adding both the parts, we get the desired claim.
\end{proof}	
	
\subsection{Proof of cost bound for optimal solution of ALP}
\label{feasiblesolution-ALP}
    
    
    {\em Feasibility:} \begin{enumerate}
	\item For $j' \in \csparse$,
	$\sum_{i\in \T{j'}}{}w'_i \leq \sum_{i\in \bundle{j'}}{}\hat{x}_{ij'} \leq 1$.
	\item For $j' \in \cdense$, $\sum_{i \in \T{j'}}{} w'_{i} = \sum_{i \in \neighbor{j'}}{} \frac{d_i}{u}=\frac{\Delta_{j'}}{u}
	\geq {\floor{\frac{\Delta_{j'}}{u}}}$ as $\sum_{i \in \bundle{j'}}{} d_i = \Delta_{j'}$.
	
	\item For a \MC $\Bundler$, we have  
	$\sum_{j' \in \Sr{r}}{}~\sum_{i \in \T{j'}}{} w'_{i}
    = \sum_{j' \in \cen{r}\cap \csparse}{}~\sum_{i \in \T{j'}}{} \hat{x}_{ij'}
    \geq \sum_{j' \in \cen{r} \cap \csparse}{} (1 - \frac{1}{\ell})^2 = \sigma_r(1-\frac{1}{\ell})^2 \geq \max\{0,\sigma_r-1\} = \alpha_r$ where the last inequality follows for $\sigma_r \leq \ell/2$ and $\ell \ge 4$.
	
	\item 
	Also, $\sum_{i \in \facilityset}{} w'_{i}  \leq \sum_{i \in \facilityset}{} \hat{y}_{i} \leq k$. 
	
	\end{enumerate}
	
	{\em Cost Bound:} Next, consider the objective function.
	
	\begin{enumerate}
	    \item Consider facility opening cost. Clearly,	$\sum_{i \in \facilityset}{} f_i w'_{i}  \leq  \sum_{i \in \facilityset}{} f_i \hat{y}_{i} \leq LP_{opt} $.
 
 \item Next, consider the part of objective function for $\cdense$. 
 
 For $j' \in \crich$, we have,
 $\sum_{i \in \T{j'}} \capacity~\dist{i}{j'} w'_i \\= u\sum_{i \in \bundle{j'}} \dist{i}{j'}(\frac{\sum_{j \in C_r} \hat{x}_{ij}}{\capacity}) \\= \sum_{i \in \bundle{j'}} \sum_{j \in C_r} \dist{i}{j'} \hat{x}_{ij}\\ \leq \sum_{i \in \bundle{j'}} \sum_{j \in C_r} \dist{i}{k'} \hat{x}_{ij}$ 
 (as $i$ belongs to $\bundle{j'}$ and not $\bundle{k'}$ for some $k' \ne j'$) 
 $\\ \leq \sum_{i \in \bundle{j'}} \sum_{j \in C_r} \left( \dist{i}{j} + \dist{j}{k'} \right) \hat{x}_{ij}$ (using triangle inequality) 
 $\\ \leq \sum_{i \in \bundle{j'}} \sum_{j \in C_r} \left( \dist{i}{j} + 2\ell \C{j} \right) \hat{x}_{ij}$ because( $\dist{j}{k'} \leq 2\ell\C{j}$)
 
Summing over all $j' \in \cdense$ we get, 
	$\sum_{j' \in \crich}{} \sum_{i \in \bundle{j'}}{} \sum_{j \in C_r}{} \hat{x}_{ij} \lbrack \dist{i}{j} +2\ell\C{j} \rbrack \leq (2\ell+1) \cdot  CostCkFLPP(\hat{\rho})$.
	
	\item Now consider the part of objective function for $\csparse$. For $j' \in \crich$, we have, 
	
	$\sum_{j' \in \csparse }{} ~ \demandofj{j'} [\sumap{i \in \bundle{j'}}\dist{i}{j'} w'_{i} + \dist{j'}{\sigmaone(j') }  ( \frac{\ell-1}{\ell} - min \{ \frac{\ell-1}{\ell}, \sum_{i \in \bundle{j'}}{} w'_{i})] \} $
	
	$= \sum_{j' \in \csparse }{} ~ \demandofj{j'} [\sum_{i \in \bundle{j'}}\dist{i}{j'} \hat{x}_{ij'} + \dist{j'}{\sigmaone(j') }  ( \frac{\ell-1}{\ell} - min \{ \frac{\ell-1}{\ell}, \sum_{i \in \bundle{j'}}{} \hat{x}_{ij'})] \}$
	
	$\leq 12 \cdot CostCkFLPP(\hat{\rho})$ by Lemma \ref{factor12}.
	\end{enumerate}

	Thus, the solution $w'$ is feasible and $CostCkFLPP(w')$,
	
	$\sum_{j' \in \csparse}{} ~ \demandofj{j'} ~\left[ \sum_{i \in \bundle{j'}}\dist{i}{j'} w'_{i} + \dist{j'}{\sigmaone(j')}  \left( \frac{\ell-1}{\ell} - min\{ \frac{\ell-1}{\ell}, \sum_{i \in \bundle{j'}} \hat{x}_{ij'} \} \right) \right] + \\ \capacity \sum_{j' \in \crich} \sum_{\singlefacility \in \bundle{j'}} \dist{i}{j'} w'_i + \sum_{i \in \facilityset} f_iw'_i
	\leq (2\ell+14) \cdot  CostCkFLPP(\hat{\rho})$.

\subsection{$(2+\epsilon)$ factor violation in capacities}
\label{app_(2+e)}
\begin{lemma}
\label{feasiblesolution-costKNM2}
$w'$ is a feasible solution to ALP with new constraints and $ CostALP(w') \leq (2\ell+14) \cdot  CostCkFLPP(\hat{\rho})$.
\end{lemma}
  
 \begin{proof}
 The feasible solution $w'$, feasibility of the constraints other than Constraints~\ref{ALPnew-1} and \ref{ALPnew-2} and the cost bound is same as that in Appendix~\ref{feasiblesolution-ALP}.
 
 \begin{enumerate}
     \item Consider Constraints~\ref{ALPnew-1},  $\sum_{j' \in \bundleone}~\sum_{i \in \T{j'}} w'_{i} = \sum_{i \in \T{j_b}} w'_{i} = \sum_{i \in \T{j_b}} \frac{d_i}{u} = \frac{\Delta_{j_b}}{u} \geq \floor{\frac{\Delta_{j_b}}{u}}$ when $(\frac{\Delta_{j_b}}{\capacity} -\floor{\frac{\Delta_{j_b}}{\capacity}}) \leq \delta = \alpha_{r}^b$. Otherwise, $\sum_{j' \in \bundleone}~\sum_{i \in \T{j'}} w'_{i} = \sum_{i \in \T{j_b}} w'_{i} + \sum_{i \in \T{j_s}} w'_{i} = \sum_{i \in \T{j_b}} \frac{d_i}{u} + \hat{x}_{ij_s} \geq \frac{\Delta_{j_b}}{u} + (1-\frac{1}{\ell})^2 \geq \floor{\frac{\Delta_{j_b}}{u}} + \delta + (1-\frac{1}{\ell})^2 \geq  \floor{\frac{\Delta_{j_b}}{u}} + 1$ for $\ell \geq \frac{2}{\delta} = \alpha_{r}^b$.
     
     \item Consider Constraints~\ref{ALPnew-2}, if $(\frac{\Delta_{j_b}}{\capacity} -\floor{\frac{\Delta_{j_b}}{\capacity}}) \leq \delta$, then $\bundletwo=\Bundler \cap \csparse$. 
     
     Therefore, $\sum_{j' \in \bundletwo}~\sum_{i \in \T{j'}} w'_{i} \geq \sum_{j' \in \cen{r} \cap \csparse}{} (1 - \frac{1}{\ell})^2 = \sigma_{r}^s(1-\frac{1}{\ell})^2 \geq \max\{0,\sigma_{r}^s-1\} = \alpha_r^s$ where the last inequality follows for $\sigma_r^s \leq \ell/2$ and $\ell \ge 4$. Otherwise, $\sum_{j' \in \bundletwo}~\sum_{i \in \T{j'}} w'_{i} = \sum_{j' \in \bundletwo}~\sum_{i \in \T{j'}} \hat{x}_{ij'} \geq \sum_{j' \in \bundletwo} (1-\frac{1}{\ell})^2 = \sigma_r^s(1-\frac{1}{\ell})^2 \geq \max\{0,\sigma_{r}^s-1\} = \alpha_r^s$ where the last inequality follows because $\sigma_r^s \leq \ell/2$.
 \end{enumerate}\end{proof}
 
 An integral solution $\bar{w}$ is obtained in a similar manner as in Section~\ref{ALP1}.

\end{document}